\documentclass[12pt]{article}
\usepackage{stackengine}
\usepackage{multicol, eurosym}
\usepackage[utf8]{inputenc}
\usepackage[T1]{fontenc}
\usepackage{multirow,array,caption}
\usepackage{lmodern} 

\usepackage{geometry}

\DeclareFontFamily{OT1}{pzc}{}
\DeclareMathAlphabet{\mathpzc}{OT1}{pzc}{m}{it}
\usepackage{setspace,graphicx}
\usepackage{amsmath,amsthm}
\usepackage{amssymb}
\usepackage{comment}
\usepackage{hyperref}

\usepackage[table]{xcolor}
\usepackage{booktabs}
\usepackage[authoryear]{natbib}
\newtheorem{proposition}{Proposition}

\def\argmax{\mathrm{argmax}}
\def\argmin{\mathrm{argmin}}
\def\Re{\mathbf{R}}
\def\ep{\varepsilon}
\def\sa{\sigma}
\def\Avg{\textrm{Avg}}
\def\PC{\mbox{P\&C} }
\def\ta{\theta}
\newcommand{\df}[1]{\textit{#1}}

\begin{document}

\title{\textsc{Price \& Choose}\thanks{We thank Danilo Coelho, Yukio Koriyama, Jean-François Laslier, Olivier Tercieux, Hal Varian, Rodrigo Velez and Dimitrios Xefteris for their useful remarks and comments. Mat\'ias N\'u\~nez  is supported by a grant of the French National Research Agency (ANR), "Investissements d'Avenir" (LabEx Ecodec/ANR-11-LABX-0047). }}
\date{ \textsc{March 2022}}
\author{Federico Echenique\thanks{\href{mailto:fede@econ.berkeley.edu}{fede@econ.berkeley.edu}. UC Berkeley .} \and Matías Núñez\thanks{\href{mailto:matias.nunez@polytechnique.edu}{matias.nunez@polytechnique.edu}. CREST \& Ecole Polytechnique.}}
\maketitle

\begin{abstract}
We describe a two-stage mechanism that fully implements the set of efficient outcomes in two-agent environments with quasi-linear utilities. The mechanism asks one agent to set prices for each outcome, and the other agent to make a choice, paying the corresponding price:  Price \& Choose.  We extend our  implementation result in three main directions: an arbitrary number of players, non-quasi linear utilities, and robustness to max-min behavior. Finally, we discuss how to reduce the payoff inequality between  players while still achieving efficiency.
\end{abstract}

	
\medskip
	

\section{Introduction}

Ted and Joanna Kramer are getting a divorce. It gets messy. Not only do they need to divide their assets, but they must decide on a complex arrangement for custody and visitation rights of their son Billy. Many possible outcomes are on the table, and the Kramers agree to use an outside arbitrator to find a solution.  The arbitrator, Judge Atkins, does not know the Kramers' preferences over the different possible solutions, but wants to find a good compromise despite his ignorance.

Our paper proposes a simple solution for Judge Atkins and the Kramers. The solution is optimal, in the sense of producing an efficient outcome in any equilibrium of the ensuing game of Kramer vs.\ Kramer, and relies on two key aspects of the problem: First, that Joanna and Ted know each other very well. Their preferences are common knowledge between them. Judge Atkins does not know the Kramers' preferences, but can leverage their shared knowledge. Second, Ted has a high-paying job as an advertising executive; so they have money available to facilitate an agreement. 

Our solution is a simple dynamic mechanism that we call \textit{Price and Choose} (\PC in the sequel). It works as follows:\begin{enumerate}
    \item The first mover sets up a zero-sum price vector that specifies a price for each of the different options.
    \item The second mover chooses one of the options as the outcome, and pays the first mover the specified price.
\end{enumerate}
We prove that any equilibrium outcome of \PC is Pareto efficient. The intuition behind our result is straightforward: The first mover's best choice is to make the second mover indifferent among all options; otherwise she is not playing optimally, as she can improve by slightly altering the price vector without modifying the choice of the second mover. This indifference implies that the second mover obtains her average utility across the options in equilibrium. The second part of our argument shows that the second mover chooses the best option(s) for the first mover. A different choice could not emerge in equilibrium, as the first mover would ``punish'' her by slightly modifying prices. This ends our argument  since the option that maximizes the payoff of the first mover, including transfers, is necessarily one that maximizes the sum of the utilities: an efficient option. 

Our paper contributes to the general theory of implementation, and to the more practical literature on arbitration. We proceed to discuss each of these connections in more detail.

Arbitration is a private dispute resolution method that does not involve courts. While the model in our paper is quite general, the problem faced by arbitrators, such as Judge Atkins, is a good application of several aspects of the model we develop.\footnote{Arbitration is not the only application of our work; note that our results extend to a setting with an arbitrary number of players so that agreements among countries or firms is also a good illustration of the current results.}  Our method allows the (two) involved parties to choose the arbitrator who will resolve the dispute.\footnote{As argued by \cite{barbera2022compromising}, practically all cross-border commercial disputes are resolved by arbitration.} These institutions may specify a structured selection procedure to help the parties exercise their right of choice, such as the American Association of Arbitrators (for instance, using vetoes, points, etc.). Two recent papers have proposed methods to improve the procedures used by practitioners. The first one, \cite{declippel2014}, proposes a ``shortlisting'' mechanism. Shortlisting works in only two stages, and the paper tests its validity in the lab. The second paper, \cite{barbera2022compromising}, considers procedures with more steps, but achieving less inequality among players. A key common ingredient in the problems studied in these two papers is that they do not allow for monetary transfers between the players. 

The lack of transfers is a realistic feature of some problems, but not of others. For problems like the Kramers', it makes sense to assume that money is available, and that it may be used to facilitate an agreement.  Economic theory has shown that introducing transfers (or prices) can serve as a powerful coordination tool and lead to welfare gains. Our proposal relies on transfers, but can accommodate rather general preferences over money. In the paper, we first consider a setting where preferences are quasilinear, and two players need to reach an agreement. Our goal is to design mechanisms that implement the utilitarian goal (that is, players end up maximizing the sum of the individual utilities). With quasilinear preferences, utilitarianism captures economic efficiency exactly. Then we generalize the result to a setting with separable, but non-quasilinear, preferences over transfers. In consequence, our results allow for example, for general attitudes towards risky monetary lotteries. 


Next, we turn to a discussion of implementation. Implementation theory studies procedures for collective decision-making in the presence of selfish agents who may disagree on their preferences over outcomes. So-called full implementation looks for procedures that induce a desirable outcome, regardless of equilibrium selection. It is often difficult to achieve when there are only two agents, as in the example with Joanna and Ted. Our paper considers full (subgame perfect) implementation in a general social choice problem with monetary transfers. The \PC mechanism we propose has the benefit of being natural and bounded, in contrast with some well-known proposals in the literature on implementation that rely on integer games and unbounded message spaces to rule out equilibria (see \cite{jackson1992implementation} for a critical review). Implementation is often challenging when there are only two agents, but our baseline analysis of the \PC applies precisely to the model with two agents. In fact, the extension to $n$ agents works by recursively applying our result for two agents.

The literature on implementation with transfers is not new. The classic demand-revealing mechanisms (see \cite{clarke1971multipart} and \cite{groves1976information}) achieve implementation in dominant strategies, even though they fail to be budget-balanced. These mechanisms require utility to be quasi-linear in transfers. Our mechanism, in contrast, achieves full implementation in subgame-perfect equilibrium; but it is budget balanced, and does not require quasilinear utility. \cite{groves1977optimal} describe a mechanism that yields efficient Nash equilibria for the public-goods problem, see \cite{groves1979efficient} for an excellent summary. The more recent literature has shifted its attention to simple mechanisms: \cite{varian1994solution} designs compensation mechanisms that achieve efficiency in the presence of externalities. Such mechanisms are not balanced off-equilibrium, whereas the \PC mechanism is balanced by definition.\footnote{The compensation mechanisms in \cite{varian1994solution} rely on fines to ensure that both players accurately report each other's ``type,'' which pushes transfers to be balanced in equilibrium; with three players and more, compensation mechanisms rely on classical implementation ideas to make each player's payment do not depend on his own report. The \PC mechanism does not depend on this logic since it gives each player either the possibility of setting a price vector (except the last one) which balances the transfers.} Similarly, \cite{jackson1992implementing} describes simple mechanisms that implement efficient allocations in undominated Nash equilibria; yet, the implementation result only applies with indivisible public goods and quasi-linear utilities. As we have mentioned above, our results extend beyond this setting.

The main result is stated in a stylized setting, but our arguments turn out to generalize in various ways. Specifically, we show that the \PC mechanism can be adapted, and efficiency still implemented, in the following variations of our basic model:
\begin{itemize}
    \item $n$ players. Moving in order, all players but the last one, choose  a price vector that the next player faces. Prices must add to zero across outcomes. The last player chooses an outcome, say $x$. Then each player pays their predecessor the price that they demanded for $x$. Here the first mover receives a transfer but does not make any, the last mover makes a transfer but does not receive any, whereas each of the other players receives and makes transfers.   Our two-player result can be applied ``recursively'' to show that the mechanism implements the efficient options. (Section~\ref{sec:manyplayers})
    \item Non-quasi linear preferences. We consider a model in which agents have general additively separable preferences over money and outcomes, and show that the main result of the paper continues to hold. (Section~\ref{sec:nql})
    \item Robust implementation. We relax the assumption that players play an exact subgame-perfect Nash equilibrium. Instead, the agents are $\ep$-maximizing, and one player makes a pessimistic worst-case assumption over the possible $\ep$-optimizing choices of the other player. (Section~\ref{sec:robust}) 
    \item Endogenous order of play. We tackle the implied first-move advantage in \PC by having players bid for the role of moving first. (Section~\ref{sec:bidpandc}). There is an alternative approach to dealing with the first-mover advantage by constraining prices to add up to a non-zero constant. This is briefly discussed after our basic result is stated. 
\end{itemize}

The rest of the paper is organized as follows. Section \ref{sec:review} reviews the literature. After laying down the model in Section \ref{sec:setting}, Section \ref{sec:pandc} presents the Price \& Choose mechanism for two players and presents the implementation argument. Sections \ref{sec:manyplayers} and \ref{sec:nql} respectively extend the model to an arbitrary number of players and to non quasi-linear utilities. Section \ref{sec:robust} presents the robustness of the mechanism with respect to adversarial behavior. Finally, Section \ref{sec:bidpandc} discusses the other previously mentioned technical extensions.
 
 \section{Review of the literature \label{sec:review}}
 
Classical results in implementation say that, with two players, and in the absence of transfers, the only Pareto efficient rule that is implementable is dictatorship (see \cite{maskin1999nash} and \cite{hurwicz1978construction}). While more permissible results arise when domains are restricted (\cite{moore1990nash} and \cite{dutta1991necessary}), or when mechanisms are not deterministic (\cite{laslier2021solution}), a commonly held view is that it is hard to design mechanisms with desirable properties in two-player settings. This has led the literature to consider the short mechanisms (short in the sense of few steps) proposed by \cite{declippel2014} and \cite{barbera2022compromising}. 

Our \PC mechanism deviates from these papers by working in an environment with transfers, and by being dynamic in nature---one player sets up a price and the other player chooses an option. The resulting solution concept is subgame-perfect Nash equilibrium. One strand of the literature is concerned with the design of mechanisms with transfers. Beyond the papers previously cited,  \cite{hurwicz1977dimensional}, \cite{dutta1994nash}, \cite{sjostrom1994implementation} and \cite{saijo1996toward} study Nash implementation when players announce prices and quantities. Among other findings, they prove that the no-envy and Pareto correspondence are implementable. 
\cite{moore1988subgame} prove that in the quasi-linear setting, any social choice rule is implementable with two players. Yet, this result has been subject to several criticisms (see \cite{aghion2012subgame}); and \cite{moore1988subgame} themselves write that their mechanisms ``are far from simple; players move simultaneously at each stage and their strategy sets are unconvincingly rich.'' Our \PC mechanism is arguably very simple, and uses a natural economic framework. It also continues to work, even when we deviate from the quasi-linear setting (Section~\ref{sec:nql}).\footnote{Our model with transfers, but non-quasi-linear preferences, is related to the recent literature on  matching problems with imperfectly transferable utility such as \cite{legros2007beauty}, \cite{chiappori2016econometrics} and \cite{galichon2019costly}.}

A literature in social choice theory (see \cite{green93}, \cite{chambers2005multi}, and \cite{chambers2005additive}) considers similar environments to ours, and studies efficient solution axiomatically. This work is, however, not concerned with implementation.

In dynamic environments with transfers, the \PC mechanism is also related to  Gary Becker's ``Rotten Kid theorem;'' see \cite{bergstrom1989fresh} for a formal analysis. Bergstrom  shows how to achieve efficiency in the Rotten Kid two-stage game, where a benevolent planner makes transfers to several selfish players. His main result involves quasi-linear preferences, but also discusses extensions that do not involve these preferences. To cite the most relevant of them, \cite{chen2018getting} considers two-stage stochastic mechanisms that achieve full implementation  under  initial  rationalizability  in  complete  information  environments. \cite{chen2021maskin} consider implementation allowing for lotteries and monetary transfers in the mechanism and characterize the implementable rules. This is, of course, different from \PC which is not a random mechanism. 

Given our motivation, we should mention the literature on dissolving partnerships: for example \cite{cramptongibbonsklemperer}. The literature is usually focused on mechanism design, not full implementation, and considers more restrictive environments than we have studied here. \cite{CRAWFORD197910} and \cite{mcafee1992amicable} consider variations of the ``cut and choose'' mechanisms, which were an inspiration of sorts for our mechanism. The name ``price and choose'' is meant to highlight this connection. Of course, cut and choose (or divide and choose) make sense for allocation problems, not for the general social choice environments we have studied here. Among the literature on cake-cutting, the work of \cite{nicolo2017divide} is relevant; they consider sequential Divide-and-compromise rules under monetary transfers, which share some common ideas with P\&C, but are designed for situations where a collectively owned indivisible good is to be divided between two agents. Furthermore, \cite{brown2016costs} study these procedures experimentally, and find evidence that second movers tend to be adversarial, an assumption related to our extension in section \ref{sec:robust}.

Finally, we should also mention the literature that crafts mechanisms  implementing efficient options, such as \cite{perez2002choosing}, \cite{ehlers2009choosing} and \cite{eguia2021implementation}. The common feature of the mechanisms designed by these papers is that they are simultaneous, and rely on lotteries as tie-breaking devices. Our approach differs from theirs in that we design a deterministic dynamic (with sequential choices) mechanism. We think of this distinction as an advantage. On the one hand, \cite{declippel2014} and \cite{camerer2016evaluating} have shown that subgame perfect equilibrium is a good predictor in the lab for a particular mechanism, namely shortlisting. The shortlisting mechanism is closely related to \PC, since the first-mover proposes a list of alternatives, and her opponent selects an alternative from the proposed list. On the other hand, Nash equilibrium often performs poorly in experimental designs with simultaneous interactions and lotteries. 

 \section{The Model \label{sec:setting}}
 
 \textbf{Utilities.}  We consider a finite set $N$ of players, with generic element $i$, who bargain over a finite set of \df{options} denoted by $A=\{a_1,a_2,\ldots,a_k\}$. Players have quasi-linear utility functions, defined over the options in $A$ and money. So player $i$ has a utility function $u_i:A\to \Re$, and enjoys a utility of $u_i(a)+t_i$ if the outcome is $a\in A$ and they receive a monetary transfer $t_i$.\footnote{The assumption of quasi-linearity is relaxed in Section~\ref{sec:nql}.}

 For each player $i$, we write $\Avg_i$ to denote $\frac{1}{k}\sum_{j=1}^k u_i(a_j)$, the average utility over $A$ for player $i$. Utilitarian welfare from an option $a$ is $\sum_i u_i(a)$, and $\textsc{max}(u)= \max
\{\sum_{i=1}^{n}u_i(a_j):1\leq j\leq k \}$ denotes the maximum utilitarian welfare.  

 \noindent \textbf{Outcomes.} For each $i$, $t_i$ denotes the monetary transfer that player $i$ obtains and $t=(t_1,\ldots,t_n)\in \mathbf{R}^n$ is a vector of transfers. An \df{allocation}, or \df{outcome}, $(a,t_1,\ldots,t_n)$ is a decision (that is, an option in $A$) coupled with a vector of transfers. 
 
 \noindent \textbf{Welfare.} An allocation $(a,t_1,\ldots,t_n)\in A\times \Re^n$ is \df{Pareto optimal} if there is no other allocation $(a',t'_1,\ldots,t'_n)\in A\times \Re^n$ with 1) $u_i(a)+t_i\leq u_i(a')+t'_i$ for all $i$, 2) $u_i(a)+t_i< u_i(a')+t'_i$ for all least one $i$, and 3) $\sum_{i=1}^n t'_i\leq \sum_{i=1}^n t_i$. An outcome $a\in A$ is \df{efficient} if $ \textsc{max}(u)= \sum_{i=1}^n u_i(a)$, so it achieves maximum utilitarian welfare.   It is well known that an allocation $(a,t_1,\ldots,t_n)\in A\times \Re^n$ is Pareto optimal if and only if $a$ is efficient.
 
\noindent \textbf{Subgame perfect implementation.}  We provide an informal definition of subgame-perfect implementation because the paper is devoted to a particular mechanism, so providing a formal and general definition is a big distraction. 

A mechanism specifies a game-form: this means that, when the mechanism is coupled with utility functions over outcomes for each of the players, it defines an extensive-form game. For a mechanism $\ta$, let $\mathrm{SPNE}^{\theta}(u)$ be the set of subgame perfect equilibria when the utility profile is $u$. A mechanism subgame perfect implements the set of efficient options if for any $u$, any member of $\mathrm{SPNE}^{\theta}(u)$ selects an efficient option and any efficient option is selected by some member of $\mathrm{SPNE}^{\theta}(u)$. 

 \section{Price \& Choose mechanism\label{sec:pandc}}
 
We proceed with our baseline result by first describing the Price \& Choose mechanism for two players, and showing that it achieves efficient implementation in subgame-perfect Nash equilibrium.
 
Consider an instance of our model with two players. The Price \& Choose mechanism requires player 1 to set up a price vector, that is a price for each option. Prices may be positive or negative, and ``budget balanced,'' in the sense of having to add up to zero. Player 2 then chooses an alternative in $A$, and pays player 1 the price that she demanded for that alternative. 

The formal definition of the Price \& Choose mechanism (P\&C) follows. Let $P=\{p\in\Re^{|A|}:\sum_{j=1}^k p_j =0 \}$

 \noindent\textbf{Timing.}

\medskip

\noindent 1. Player 1 chooses a price vector $p\in P$.

\noindent 2. Player $2$ chooses an option $a\in A$ and transfers $p_a$ to Player 1.

\medskip
 
For any option $a$ chosen at the second stage and any price vector $p$ set in the first stage, the payoffs associated to this mechanism equal $g(p,a)=(g_1(p,a),g_2(p,a))=(u_1(a)+p_a,u_2(a)-p_a)$.

Now it is obvious that the \PC mechanism defines an extensive form game, given the players' utility functions. A strategy profile in the game induced by the \PC mechanism is a pair $\sigma=(\sigma_1,\sigma_2)$, with $\sigma_1\in P$ and $\sigma_2:P\to A$.  It is also obvious that there exists at least one subgame perfect Nash equilibrium in pure strategies, as this is a finite perfect-information game.  

We say that the \PC mechanism \df{implements the efficient options in subgame-perfect equilibrium} if, for any subgame-perfect Nash equilibrium $\sa=(\sa_1,\sa_2)$, $\sa_2(\sa_1)$ is efficient; and, conversely, for any efficient outcome $a\in A$, there is a  subgame-perfect Nash equilibrium $\sa=(\sa_1,\sa_2)$ with $a=\sa_2(\sa_1)$.

\begin{proposition} \label{prop:basic} \PC subgame-perfect implements the set of efficient options.
\end{proposition}

\begin{proof}
The proof is divided in four steps: the existence of a price vector that makes Player 2 indifferent between all options (Step A), the proof that this price vector is the unique one compatible with equilibrium behavior (Step B), the proof that any equilibrium selects an efficient option (Step C) and finally (Step D) the construction of an equilibrium selecting an efficient option and the converse construction: for each efficient option, there is an equilibrium implementing it. 

\noindent \textbf{Step A:} there is one and only one price $p^*\in P$ with $g_2(p^*,a_j)=u_2(a_j)-p^*_j$ being constant in~$j$. 

Indeed if $\theta= u_2(a_j)-p_j$, then \[ 
k \theta = \sum_{j=1}^k u_2(a_j)- \sum_{j=1}^k p_j = k \Avg_2,
\] as $p\in P$. Therefore, $p^*_j = u_2(a_j) - \theta = u_2(a_j) - \Avg_2$.

\noindent\textbf{Step B:} If $\sigma$ is a subgame-perfect Nash equilibrium, then $\sigma_1=p^*$.

Let $\sigma$ be a subgame-perfect equilibrium, $p=\sigma_1$ and $a_i=\sigma_2(p)$. We claim that $g_2(p,a_j)=u_2(a_j)-p_j$ is constant in~$j$. Suppose then, towards a contradiction, that there is $j$ and $h$ with $g_2(p,a_h)>g_2(p,a_j)$. Let $H$ be the set of $h$ with $a_h\in \argmax\{u_2(a_j)-p_j:1\leq j\leq k \}$, and note that $i\in H$ while $j\notin H$. Consider the price vector $p'$ that is identical to $p$ except that $p'_i=p_i+\ep$, $p'_h=p_h+2\ep$ for $h\in H\setminus\{ i\}$, and $p'_j=p_j +\ep - 2\ep | H |$. For $\ep>0$ small enough, Player $2$ finds it uniquely optimal to choose $a_i$, while player 1's payoff is strictly greater. A contradiction.

Since  $g_2(p,a_j)=u_2(a_j)-p_j$ is constant in~$j$, by Step A, $p=p^*$.

\noindent \textbf{Step C:} If $\sigma$ is a subgame-perfect Nash equilibrium, then $\sigma_2(p^*)\in \argmax \{u_1(a_j)+u_2(a_j):1\leq j\leq k \}$.

Suppose, towards a contradiction, that $\sigma_2(p^*)=a_j$ and that $u_1(a_j)+u_2(a_j)<u_1(a_i)+u_2(a_i)$. By definition of $p^*$, however, $u_2(a_j)-p^*_j = u_2(a_i)-p^*_i$. Suppose now that player 1 chooses a price vector $p'\in P$ that is identical to $p^*$, except in that $p'_i=p_i^*-\epsilon$ and $p'_j=p^*_j+\epsilon$, for $\epsilon>0$. Then we have that $\sigma_2(p')=a_i$, as now $a_i$ is the uniquely optimal choice for player 2, while
\[ u_1(a_j) + p^*_j = u_1(a_j)+u_2(a_j) - \Avg_2 < 
u_1(a_i)+u_2(a_i) -\epsilon - \Avg_2 = u_1(a_i) + p'_i,
\] for $\epsilon>0$ small enough, contradicting that $\sigma$ is a subgame-perfect Nash equilibrium.



\noindent \textbf{Step D:} For every efficient outcome, there is a subgame-perfect Nash equilibrium that selects it.

Observe that Step A is a general remark on the mechanism whereas Steps B and C
deal with any subgame perfect equilibrium. This means that in any subgame perfect equilibrium: $\sigma_1=p^*$ and $\sigma_2(p^*)\in \argmax \{u_1(a_j)+u_2(a_j):1\leq j\leq k \}$. In other words, any subgame perfect equilibrium outcome is efficient. Let $a_j$ be an efficient outcome. Consider the strategy profile $\sigma_1=p^*$ and $\sigma_2(p^*)=a_j$. Player 2 is playing a best response since $p^*$ is making him indifferent between all options. 
Player 1's payoff equals :

$$u_1(a_j)+p^*_j=u_1(a_j)+u_2(a_j)-\Avg_2=\textsc{max}(u)-\Avg_2.$$

This means that, given $p^*$, Player 1 is indifferent among all efficient options. Now, suppose that Player 1 alters the price vector to force Player 2 to choose another efficient option. Any option $a_h$ with a price $p_h>p^*_h$ will not be chosen by Player 2. This means that if Player 1 wants to induce Player 2 to choose some option $a_l$ he must set $p_l<p^*_l$; however, this implies that if Player 2 chooses $a_l$ Player 1's payoff is lower or equal than $\textsc{max}(u)-\Avg_2$ since Player 1's payoff is increasing on $p_l$, a contradiction. Hence, any efficient option is selected in some subgame perfect equilibrium.

Observe that, in Step D, the existence of a subgame-perfect Nash equilibrium is established.
\end{proof}

The \PC mechanism confers
\color{black}the first player an advantage, as the equilibrium payoffs to Player 2 are always $\Avg_2$, while Player 1 gets a payoff that is greater than $\Avg_1$. We consider this issue in detail in Section~\ref{sec:bidpandc}, but we note here that the assumption that prices in $P$ add to zero may be modified to avoid (or exacerbate) the payoff imbalance. 

Indeed, for any constant $\alpha$, we may define the \PC mechanism with Player 1 choosing a price in $P_\alpha=\{p\in \Re^{|A|}: \sum_{j=1}^k p_j =\alpha \}$.  By considering a modified game, with $P=P_0$ as above, but in which Player 2's utility is $u_2-\alpha/k$, and 1's utility is $u_1+\alpha/k$, we see that \PC again subgame-perfect implements the efficient alternative. Now, however, Player 2's payoff is $\Avg_{2}-\frac{\alpha}{k}$ while 1's payoff is $\textsc{max}(u)-\Avg_{2}+\frac{\alpha}{k}$. A negative value of $\alpha$ serves to balance the payoffs to the two agents.


An outside agent like Judge Atkins, who does not know the utilities of players 1 and 2, may want to use $P_\alpha$ in order to balance the  \PC mechanism, but not know the proper value of $\alpha$. It is, however, possible to endogenize the needed value of $\alpha$. One idea is to proceed as follows:
\begin{enumerate}
    \item Player 1 proposes a real number $\alpha$.
    \item Player 2 decides between being the  chooser (so that Player 1 is the proposer) or the proposer (and Player 1 becomes the chooser).
    \item The proposer set-up a price vector with $p\in P_\alpha$ and
    \item The chooser selects an alternative $a_j$ and pays $p_j$ to the proposer.
\end{enumerate}

By replicating the arguments in Proposition~\ref{prop:basic}, one can show that, in equilibrium, $\alpha=\frac{k}{2}(\Avg_1+\Avg_2-\textsc{max}(u))$ so that Player 2 is indifferent between the two roles.
This means that the respective payoffs equal $\frac{1}{2}(\textsc{max}(u))-\frac{1}{2}(\Avg_2-\Avg_1)$ and $\frac{1}{2}(\textsc{max}(u))+\frac{1}{2}(\Avg_2-\Avg_1)$. This version with an endogenous sum of the prices induces a redistribution between players with respect to the default version of \PC, in which the prices sum up to 0 and payoffs equal $(\textsc{max}(u)-\Avg_2,\Avg_2)$. The main difference between the \PC with endogenous $\alpha$ is that the payoff difference only depends on the players' average payoff, and not on the total payoff $\textsc{max}(u)$. 

Section~\ref{sec:bidpandc} fleshes out a related idea for balancing payoffs in the \PC mechanism.

\section{Price \& Choose with Many players}\label{sec:manyplayers}

We now turn to a many-player version of the problem. We shall see that the previous result implies that a simple $n$-player variation of our \PC mechanism achieves subgame-perfect implementation of the efficient options.  In this mechanism, the first $n-1$ players propose, one after the other, a balanced price vector that they demand as payment from the next player in the order. The $n$th player chooses an option $a\in A$.  The endogenously set prices determine the transfers made between consecutive players. A balanced price vector remains a vector of prices such that the sum of prices equals zero.

Formally, the $P^{n-1}\& C$ mechanism works as follows :

\medskip

\noindent \textbf{Timing.}

\medskip

\noindent 1. Player 1 sets up a price vector $p^2\in P$.

\noindent 2. For each $i=2,\ldots,n-1$, Player $i$ sets up a price vector $p^{i+1}\in P$, knowing prices $p^2,\ldots,p^{i}$.

\noindent 3. Player $n$ chooses an option $a$ as the outcome given the prices 
$p^2,\ldots,p^{n}$.

\medskip

\noindent \textbf{Transfers.}

\medskip

Say that $a_j$ is the option chosen by Player $n$; this means that Player $n$ pays Player $n-1$ the price $p^n(a_j)$. In turn, Player $n-1$ has to pay the price $p^{n-1}(a_j)$ to Player $n-2$. This applies to any Player $m$ with $m=2,\ldots,n-1$, so that he pays $p^{m}(a_j)$ to player $m-1$ while receiving the transfer $p^{m+1}(a_j)$ from player $m+1$. Finally, Player $1$ receives the transfer $p^2(a_j)$ from Player 2, but makes no further payments. This entails that, assuming quasi-linear preferences, the payoffs associated to the option $a_j$ and the price vector $p=(p^2,\ldots,p^n)$ equal:

\noindent $g_n(p,a_j)=u_n(a_j)-p^{n}(a_j)$.\\
\noindent  $g_{m}(p,a_j)=u_{m}(a_j)-p^{m}(a_j)+p^{m+1}(a_j)$ for $m=2,\ldots,n-1$.\\
\noindent $g_1(p,a_j)=u_1(a_j)+p^{2}(a_j)$.

\begin{proposition} \label{prop:nplayers} The P$^{n-1}$\&C mechanism subgame-perfect implements the set of  efficient options. 
\end{proposition}

We prove Proposition~\ref{prop:nplayers} by recursively applying Proposition~\ref{prop:basic}.

\begin{proof}

We use Proposition~\ref{prop:basic} and proceed by induction. Fix a subgame-perfect Nash equilibrium $\sa=(\sa_1,\ldots,\sa_n)$. Define $p^{n+1}=p^1=(0,\ldots,0)$ so that, for any player $i$, if the option $a$ is the outcome, with the sequence of prices $p^2,\ldots,p^n$, then $i$'s payoff is $u_i(a) - p^i(a)+p^{i+1}(a).$

We claim that in, any subgame given by $(p^1,\ldots,p^{i-1})$, the outcome under $\sa$ must be such that, for $1\leq i\leq n$ :
\begin{enumerate}
    \item\label{it:pcn1} $\sa_n(p^2,\ldots,p^n)\in A$ maximizes $\sum_{j=i+1}^n u_j(a') + u_i(a')-p^i(a')$ over $a'\in A$.
    \item\label{it:pcn2} $\sum_{j=i+1}^n u_j(\sa_n(p)) -p^{i+1}(\sa_n(p)) = \mathrm{Average}(\sum_{j=i+1}^n u_j)$.
\end{enumerate}

The proof of this claim is by induction: First, by Proposition \ref{prop:basic}, the claim is true for any subgame $(p^1,\ldots,p^{n-1})$, as the resulting subgame is an instance of the two-player game with payoffs $u_{n-1}(a)-p^{n-1}(a)$ and $u_n(a)$. Second, if the claim is true for any subsequent subgame, then for any subgame $(p^1,\ldots,p^{i})$ we may consider a two-player game between player $i$ and a fictitious player $(i+1)'$. The former has  payoffs $u_i(a)-p^i(a)$ while  $(i+1)'$  payoffs' are $\sum_{j=i+1}^n u_j(a)$. By the inductive hypothesis, this game is an instance of the two-player game, and thus by Proposition~\ref{prop:basic}, the claim follows.

Now if $\sa$ is a subgame-perfect Nash equilibrium that results in the outcome $a^*$ and price sequence $p^1,\ldots,p^{n+1}$, we have that $a^*$ maximizes $\sum_{i}u_i(a)$ over $a\in A$, by property \eqref{it:pcn2}, and $p^1=0$. 
\end{proof}

\section{Non quasi-linear preferences \label{sec:nql}}

We now discuss the properties of the \PC mechanism when players' preferences are not quasi-linear in monetary transfers. In particular, suppose that when the outcome is $(a,t_1,t_2)$ then player 1's utility is $u_1(a)+ \zeta(t_1)$, and player 2's utility is $u_2(a)+\eta(t_2)$. 

Suppose that the functions $\eta,\zeta:\Re\rightarrow\Re$ are monotone (strictly) increasing, continuous and $\eta(\Re)=\zeta(\Re)=\Re$. The assumptions placed here on agents' utilities are called {\em additive separability}. Note that when  player 2  pays $p_a$ to player 1, player 2's utility decreases by $\eta(p_a)$ while player 1's utility increases by $\zeta(p_a)$. This implies that, for any option $a$ chosen at the second stage and any price vector $p$ set in the first stage, the payoffs associated to this mechanism equal $g^{\eta,\zeta}(p,a)=(u_1(a)+\zeta(p_a),u_2(a)-\eta(p_a))$.

Since $\eta$ admits an inverse function and $\eta,\zeta$ are defined over the real numbers, the problem of analyzing the subgame perfect equilibria of the \PC mechanism with payoffs given by $g^{\eta,\zeta}$ is identical to compute the subgame perfect equilibria when payoffs equal $h(p,a)=(h_1(p,a),h_2(p,a))=(u_1(a)+\beta(p_a),u_2(a)-p_a)$ with $\beta$ being monotone increasing. In other words, we may without loss set $\eta$ to be the identify function.

\begin{proposition}\label{prop:nonQL} If utilities are additively separable, \PC subgame-perfect implements the set of Pareto optimal options.
\end{proposition}

\begin{proof}
Let $(\sa_1,\sa_2)$ be a subgame-perfect Nash equilibrium.
Steps A and B from Proposition 1 remain true with non quasi-linear preferences. We denote by $p^*=\sa_1$ the price vector selected by Player 1 in equilibrium. Moreover, Player 2 is indifferent between all options given $p^*$. Indeed, as in the proof of Proposition~\ref{prop:basic}, $p^*_a = u_2(a)-\Avg_2$. 

Now we claim that the outcome $\sa_2(p^*)=a$ must be Pareto optimal in $A$. Suppose then, towards a contradiction, that there is some $a'\in A$ with $u_i(a')\geq u_i(a)$, $i=1,2$, with strict inequality for some $i$. Then, given that $\beta$ is strictly monotone increasing, we have that 
\[ 
u_1(a)+\beta(u_2(a)-\Avg_2) < u_1(a')+\beta(u_2(a')-\Avg_2).
\] Now, however,  again as in Step C of Proposition \ref{prop:basic}, Player 1 can decrease the price of $a'$ be $\epsilon$ and increase the price of $a$ by the same amount so that $a'$ is the unique best response by player 2. Player 1's payoff is then $u_1(a')+\beta(u_2(a')-\Avg_2-\ep) > u_1(a)+\beta(u_2(a)-\Avg_2)$, for $\ep>0$ small enough.

Clearly the previous argument holds regardless of which Pareto optimal outcome is chosen, and the proof of the result may be concluded along the same lines as the proof of Proposition~\ref{prop:basic}.
\end{proof}

\section{P\&C as a robust mechanism }\label{sec:robust}

The analysis so far hinges, of course, on complete information among the players, and on the assumption of equilibrium play. We now discuss two deviations from these assumptions: perturbations from the complete information assumption and approximate equilibrium with adversarial behavior.

Regarding the first robustness check, \cite{aghion2012subgame} study implementation with transfers under perturbations of the complete information assumption. They prove that any mechanism that subgame perfect implements a social choice function  which fails to be Maskin monotonic\footnote{Maskin monotonicity plays a central role in implementation theory since any Nash implementable social choice function needs to satisfy it.} under complete information admits a sequential equilibrium with undesirable outcomes when information is perturbed. Their result is stated for finite strategy sets as well as for countably infinite ones, so it does not exactly apply to \PC and our setting. For completeness, however, we now show by example that the social choice function that the \PC mechanism implements is not Maskin monotonic. To see why, consider the next example involving three alternatives and quasi-linear preferences over transfers. Let $A=\{a_1,a_2,a_3\}$ be the set of options and $u=(u_1,u_2)$ and $u'=(u_1,u'_2)$ denote two possible utility profiles with $u_1=u_2=(1,0,-1)$ and $u'_2=(1,-2,-1)$. The set of allocations equals
$\{(a',t_1,t_2)\in A\times \mathbf{R}^2 \text{ with } t_1+t_2=0\},$
where $a'$ is the implemented option and $t_1$,$t_2$ are the transfers of players 1 and 2 respectively. The sum of the transfers is zero since Player 2 makes a transfer to Player 1 in the \PC mechanism. Recall that a social choice function $f$ maps the set of utility vectors $U$ into the set of allocations. A social choice function is Maskin monotonic on $U$ if for any pair of utility vectors $u$, $u'\in U$, if $x=f(u)$ and
$$\{(i,y)\mid u_i(x)\geq u_i(y)\}\subseteq \{(i,y)\mid u'_i(x)\geq u'_i(y)\},$$
(i.e. no player ranks $x$ lower when moving from $u$ to $u'$) then $x=f(u')$. 
The allocation $x$ chosen by the $\PC$ mechanism where 1 is the first-mover equals $x=(a_1,1,-1)$ since $a_1$ is the efficient option and the price vector is $p=(1,0,-1)$. When moving from $u$ to $u'$, no player ranks $x$ lower since $(i)$ the utility function of 1 remains unchanged and $(ii)$ for player 2, the utilities of $a_1$ and $a_3$ remain unchanged whereas the utility of $a_2$ goes down. If the rule is Maskin monotonic, $x$ should be chosen in $u'$. Yet, the chosen allocation in $u'$ is $y=(a_1,2/3,-2/3)$ and clearly $x\neq y$, violating Maskin monotonicity. As said, however, the results of \cite{aghion2012subgame} do not strictly speaking apply to our model and it remains to see if our results suffer from the lack of robustness that they study.\footnote{Some of the main results in the paper of Aghion et al deal with direct revelation mechanisms, so they are focused on mechanisms that naturally differ from \PC.}

We now consider a different robustness test in the spirit of robust mechanism design (see \cite{carroll2019robustness} for an excellent overview). The deviations we have in mind relax the notion of equilibrium in two ways: First, players are only approximately optimizing; they are ``$\ep$-maximizers.'' Second, the assumption of approximate optimization gives rise to ambiguity in how the second player will choose, and we assume that Player 1 operates under a worst-case scenario. So Player 1 expects that the ambiguity will be resolved adversarially by Player 2. As we prove below, the \PC mechanism still achieves efficient implementation in the perturbed setting we have outlined.

 The assumption that Player 2 is adversarial could be motivated by ideas of  negative reciprocity (see \cite{fehr2021behavioral} for a recent contribution on this idea in implementation), in which players' utilities depend  negatively on the utility level of their opponent. Observe also that, in the equilibrium of \PC in Section 3, the opposite behavior arises: Player 2 is indifferent between all alternatives and he chooses the one maximizing Player 1's payoff (this occurs, as we show, endogenously; it is not an assumption). 



Formally, we say that, for a fixed $\ep>0$, option $a$ is an $\ep$-maximizer for Player $2$ if there is no $a'$ that is better than $a$ by more than $\ep$. This is equivalent to saying that $ a \text{ is an } \ep\text{-maximizer for Player } 2 \Longleftrightarrow g_2(p,a)+\ep \geq g_2(p,a') \text{ for any } a'\neq a.$ We denote by $\beta^\ep_i(p)$ the set of $\ep$-maximizers at the price vector $p$ for Player $2$. The adversarial nature of Player 2 is then captured by setting $\sa_2(p) \in \argmin\{g_1(p,a) : a\in \beta_2^\ep(p) \}$. In words, Player 2 selects the option among $\ep$-maximizers that minimizes Player 1's payoff. 

Similarly, we say that Player 1 is $\ep$-maximizing when choosing a price vector $p\in P$ if $g_1(p,\sa_2(p))+\ep \geq g_1(p',\sa_2(p'))$ for all $p'\in P$. 

To sum up, we say that the strategy profile $\sa=(\sa_1,\sa_2)$ is a $\ep$-robust subgame perfect Nash equilibrium if 
\begin{enumerate}
\item $\sa_2(p)\in \argmin \{g_1(p,a):a\in\beta_2^\ep(p) \}$ for all $p\in P$,
    \item and Player 1 is $\ep$-maximizing when choosing $\sa_1\in P$.
\end{enumerate}
We say that $\sa_2(\sa_1)$ is the \df{outcome} of the $\ep$-robust subgame perfect Nash equilibrium $\sa$.

For simplicity we assume here that there is a unique efficient alternative. We expect that the argument generalizes to settings with more than one efficient option.

\begin{proposition}\label{prop:robust} For any $\ep>0$ small enough, the unique $\ep$-robust subgame perfect Nash equilibrium outcome of \PC is the efficient outcome. 
\end{proposition}

\begin{proof} 
Let $p^*$ be the price vector constructed in the proof of Proposition~\ref{prop:basic}. So  $g_2(p^*,a_j)=u_2(a_j)-p^*_j$ is constant in~$j$ and $p^*_j = u_2(a_j) - \Avg_2$ for each $a_j\in A$. 

Without loss of generality, we say that the (unique) efficient option is $a_1\in A$. So $\textsc{max}(u)=u_1(a_1)+u_2(a_1)>u_1(a_j)+u_2(a_j)$ for all $j\neq 1$.

Choose $\ep>0$ small enough so that 
\begin{equation}\label{eq:robust1}
     u_1(a_j) + u_2(a_j) +\frac{k-1}{k}\ep < u_1(a_1)+u_2(a_1) - 2\ep
\end{equation} for all $j\neq 1$.

Before we get started, observe that if $a_j=\sa_2(p)$, then $a_j\in\beta^\ep_2(p)$, and hence $u_2(a_j)-p_j+\ep\geq u_2(a_h)-p_h$ for all $h\neq j$. Therefore: 
\begin{equation}\label{eq:robust2}
    u_2(a_j)-p_j \geq \Avg_2 - \frac{k}{k-1}\ep.
\end{equation}

The proof is now divided in two steps. In Step A, we exhibit an $\ep$-subgame-perfect Nash equilibrium that selects $a_1$. In Step B, we show that, despite the potential multiplicity of equilibria, all of them select option $a_1$ as the equilibrium outcome. 

\noindent \textbf{Step A:} Consider the strategy profile $\sa$ defined by 
\begin{enumerate}
    \item[a)] $q^*=(p_1^*-\ep,p_2^*+\frac{\ep}{k-1},p_3^*+\frac{\ep}{k-1},\ldots,p_k^*+\frac{\ep}{k-1})$
    \item[b)] Player 2 chooses $a_1$ if $p=q^*$ and minimizes the payoffs of Player 1 over $\beta^2_\ep(p)$ otherwise. 
\end{enumerate}

To see why this is an $\ep$-equilibrium, observe that with $q^*$, the payoffs of both players are respectively equal to:
\begin{align*}
    g_1(q^*,\cdot) & =(u_1(a_1)+p_1^*-\ep,u_1(a_2)+p_2^* +\frac{\ep}{k-1},\ldots,u_2(a_k)+p_k^*+\frac{\ep}{k-1}) \\
    g_2(q^*,\cdot) & =(u_2(a_1)-p_1^*+\ep,u_1(a_2)-p_2^* -\frac{\ep}{k-1},\ldots,u_2(a_k)-p_k^*-\frac{\ep}{k-1}).
\end{align*}

Thus, $\beta_2^{\ep}(q^*)=\{a_1\}$, so that Player 2 chooses $a_1$ as we have claimed. To complete the proof, we need to check that Player 1 does not have a profitable deviation that exceeds their payoff by at least $\ep$. Assume, towards a contradiction, that Player 1 can find a price vector $p$ that ensures him a payoff strictly greater than $g_1(q^*,a_1)+\ep$. Let $a_j=\sa_2(p)$. Then we have $u_1(a_j)+p_j> u_1(a_1)+q^*_1+\ep$.

There are two cases to consider. The first case is when $a_j=a_1$. Then $u_1(a_1)+p_1> u_1(a_1)+q^*_1+\ep$ implies that $p_1>q_1^*+\ep = u_2(a_1)-\Avg_2$. Thus $\Avg_2 > u_2(a_1)-p_1$, and we conclude that there exists $a_j$ with $u_2(a_j)-p_j > u_2(a_1)-p_1$. 

At the same time, $a_1=\sa_2(p)$, which implies that $a_1\in\beta^\ep_2(p)$. But then  $u_2(a_j)-p_j > u_2(a_1)-p_1$  means that $a_j\in\beta^\ep_2(p)$, so $\sa_2(p)=a_1$ is only possible if $u_1(a_1)+p_1\leq u_1(a_j)+p_j$ (by the definition of $\sa_2$). Adding up these inequalities, we obtain that 
\[ 
u_1(a_1)+p_1 + u_2(a_1) - p_1 < u_1(a_j)+p_j + u_2(a_j)-p_j,
\] which contradicts the definition of $a_1$.

The second case to consider is when $a_j\neq a_1$. Then  the assumption that  $q^*$ is not an $\ep$-optimum yields that 
\[ 
u_1(a_j)+p_j> u_1(a_1)+q^*_1+\ep = u_1(a_1) + u_2(a_1) - \Avg_2.
\] Combine this inequality with  Equation~\eqref{eq:robust2} to obtain that  \[ 
u_1(a_j)+p_j  + u_2(a_j)-p_j > u_1(a_1) + u_2(a_1)  - \frac{k-1}{k}\ep,
\] contradicting~\eqref{eq:robust1}.

\noindent\textbf{Step B:} Consider any $\ep$-subgame perfect equilibrium $(p,\sa_2)$.  We claim that  $\sa_2(p)=a_1$, and suppose (towards a contradiction) that $\sa_2(p)=a_j\neq a_1$.

We first observe that $u_1(a_j)+p_j\geq u_1(a_1)+u_2(a_2)-\Avg_2-2\ep$, because Player 1 may select $q^*$ (as constructed in Step A) and guarantee a payoff of $u_1(a_j)+p_j\geq u_1(a_1)+u_2(a_2)-\Avg_2-\ep$. 

By~\eqref{eq:robust2}, we obtain \[ 
u_1(a_j)+p_j + u_2(a_j)-p_j\geq u_1(a_1)+u_2(a_2)-\Avg_2-2\ep + \Avg_2 - \frac{k-1}{k}\ep
= u_1(a_1)+u_2(a_2)-2\ep - \frac{k-1}{k}\ep,
\] contradicting~\eqref{eq:robust1}.
\end{proof}

\section{Bid, Price, and Choose \label{sec:bidpandc}}

The \PC mechanism implements a Pareto efficient option in the general model of social choice with transfers, but it does so with a particular set of transfers. In fact, the first mover is treated asymmetrically with respect to the other players.  Consider $P^{n-1}\&C$ and note that every player other than the first player in the order receives a payoff that equals $\Avg_i$, their average payoff from an option in $A$. The first moving player will receive, instead, a payoff that equals $\textsc{max}(u)-\sum_{j\neq 1}\Avg_j> \Avg_1$; a first-mover advantage. 

To correct the resulting unequal welfare distribution, we could proceed as was suggested after we stated the proof of Proposition~\ref{prop:basic}. Instead, here we focus on ideas suggested in the literature by \cite{jackson1992implementing} and \cite{perez2001bidding}; where biding in an auction determines the order of play. Specifically, all players bid to be the first mover, and the highest bidder wins (ties being broken by a uniform draw). The revenue from the winning bid is equally split among the rest of the players. Then the players play the \PC mechanism, where the player with the winning bid is the first mover. As we show, this bidding stage reduces inequality among players, and makes the equilibrium payoffs order-independent. That is, in equilibrium, players have no preferences ex-post over the stages at which to participate.

More formally, consider an auction for the role of choosing first. Each player submits a bid, $b_i\geq 0$. Let $W=\{i:b_i\equiv\max\{b_j:1\leq j\leq n \} \}$ be the set of \df{winners} --- the set of players who submitted the highest bids. One winner is chosen at random (uniformly) to pay their bid and become the first mover. The bid collected from the first mover is then distributed in equal shares among the rest of the players. So if $i\in W$ is selected, then $i$ pays $b_i$ and becomes the first mover, while all the remaining players receive a payment of  $\frac{1}{n-1}b_i$. The order of play among players who are not first is determined at random.

\begin{proposition}\label{prop:BPandC} Bid, price, and choose subgame-perfect implements the set of efficient options. Moreover, in any equilibrium, if $U^i$ are the equilibrium payoffs to players $i=1,\ldots,n$, then \[ 
U^i - U^j = \Avg_i - \Avg_j.
\] 
\end{proposition}
  
\begin{proof}
Let $\eta = \sum_{i=1}^n u_i(a^*) - \sum_{i=1}^n \Avg_i$, where $a^*$ is an efficient outcome, and
$b^*=\frac{n-1}{n}\eta$. Consider a \PC subgame, after the order of play has been determined, and observe that, in any subgame-perfect equilibrium outcome of this subgame, the payoffs to a player $j$ who is not the first mover is $\Avg_j$, while the payoff to a player $i$ who is the
first mover is $\Avg_i + \eta$. Thus, if $i$ is a winner of the auction, and is randomly chosen to move first, their payoff is $\Avg_i + \eta-b_i$. Any player $j\neq i$ gets payoff $\Avg_j+b_i/(n-1)$. By definition, $\eta-b^*=b^*/(n-1)$, so the difference in payoffs is as in the statement of Proposition~\ref{prop:BPandC}.

Note first  that there exists a symmetric Nash equilibrium of the auction with $b_i=b^*$ for all $i$, as $\eta-b^*=b^*/(n-1)$ ensures that the payoff from winning and losing are the same. Bidding higher than $b^*$ would ensure winning, but with a strictly lower payoff; and bidding lower than $b^*$ would result in losing, but getting the same payoff as with a bid of $b^*$. 

This symmetric equilibrium is not unique, but all other equilibria have the same outcome. Indeed there is no Nash equilibrium with a single winner, as the winner would gain from lowering their bid. For any $W$ with at least two players, there is a Nash equilibrium with $b_i=b^*$ for $i\in W$ and $b_i< b^*$ for $i\notin W$. This follows from the same argument as above.

Finally, consider a profile of bids with a set of winners choosing $b'\neq b^*$. At $b'$ the payoff from winning differs from the payoff from losing. If the latter is higher, a winner has an incentive to lower their bid. If the latter is lower, they can benefit by raising their bid. So there is no Nash equilibrium in which the winning bid differs from $b^*$. 
\end{proof}

\section{Conclusion}\label{sec:conclusion}

We have considered implementation in the general social choice problem with money, and an arbitrary number of agents. Our proposed solution, the Price \& Choose mechanism, is a simple procedure for reaching efficient agreements. A remarkable feature of our approach is that it relies on prices, and does not require penalties, integer games, off-equilibrium threats, or lotteries; all classical techniques used by  mechanism designers to discipline players and achieve full implementation. Our solution requires the availability of money, but does not rely on the narrow assumption of quasilinear preferences.

The main shortcoming of our approach is that it assumes complete information and equilibrium behavior; we have addressed this weakness by considering a model with maxmin and $\ep$-optimizing behavior, and shown that the set of efficient options remains implemented by the \PC mechanism - whether experimentally subjects manage to reach efficient agreements through the described methods remains an empirical question.\footnote{Another possible extension is to understand how to achieve similar results in incomplete information settings with interdependent preferences (see \cite{ollar2022efficient} for recent work in this direction).} 

\bibliographystyle{aer}
\bibliography{priceandchoose}

\appendix

\end{document}